%% file: main_arxiv.tex
\documentclass[12pt]{article}
\def \defmargin{1in}
\def \arxiv{1}
\input{macros}
\newcommand{\singcol}[2]{\ifnum\arxiv=0 {#1} \else {#2} \fi}

\begin{document}

\title{Improved Bounds for Universal\\ One-Bit Compressive Sensing}

\author{
	Jayadev Acharya\thanks{JA is supported by a start-up grant from Cornell University. Email: \texttt{acharya@cornell.edu}.} \\ Cornell University
	\and
	Arnab Bhattacharyya\thanks{AB is supported in part by a DST Ramanujan Fellowship. Part of the research was conducted while visiting the Simons Institute for Theory of Computing. Email: \texttt{arnabb@csa.iisc.ernet.in}.} \\ Indian Institute of Science
	\and
	Pritish Kamath\thanks{PK is supported in part by NSF grants CCF-1420956, CCF-1420692, CCF-1218547, CCF-1650733. Email: \texttt{pritish@mit.edu}.}\\ MIT}

\maketitle

\input{abstract}

\input{sec_intro}

\input{sec_upperbound}

\input{sec_lowerbound}

\input{sec_approx-recovery}

\input{sec_discussion}

\section*{Acknowledgment}
AB would like to thank Prateek Jain for preliminary discussions about these problems. AB would also like to thank Prateek Jain, Satya Lokam and N.V. Vinodchandran for discussions about Open Problem~\ref{open:supp-recovery-0-1}. The authors thank Russell Impagliazzo for pointing out the connection between Union Free Families and Nisan-Wigderson designs. The authors thank Jiawei Gao, Russell Impagliazzo and Madhu Sudan for discussions on explicit constructions of Robust Union Free Families (see Open Problem~\ref{open:explicit_constructions}). Finally we thank anonymous reviewers for several useful comments that improved the presentation of the paper, and also for pointing us to relevant work in the group testing literature.


\input{refs.bbl}
\appendix

\input{sec_NWdesigns}

\end{document}

%% file: macros.tex
\providecommand{\todotoggle}{1} 

\usepackage{amsmath, amssymb, amsthm, amsfonts, bbm}
\usepackage[tbtags]{mathtools}
\usepackage{graphicx}
\usepackage{verbatim}
\usepackage[usenames,dvipsnames]{color}
\usepackage{units}
\usepackage{tikz}
\usepackage{array}
\usepackage{subfigure}
\usepackage{setspace}
\usepackage{enumitem}
\usepackage[utf8]{inputenc}
\usepackage[T1]{fontenc}
\usepackage[pdfstartview=FitH, pdfpagemode=None, colorlinks=true, citecolor=blue, linkcolor=blue, urlcolor=magenta, pagebackref=true, unicode=false]{hyperref}
\usepackage[ruled,linesnumbered]{algorithm2e}
\usepackage{multirow}
\usepackage{thmtools,thm-restate}

\providecommand{\arxiv}{0}
\ifnum\arxiv=1
\usepackage{mathpazo}

\providecommand{\defmargin}{1in} 
\usepackage[top=\defmargin, bottom=\defmargin, left=\defmargin, right=\defmargin]{geometry}
\fi

\usepackage{xargs} 
\usepackage[colorinlistoftodos,prependcaption,textsize=footnotesize]{todonotes}

\newcommand{\mytodo}[1]{\ifnum\todotoggle=1{#1}\fi}

\newcommandx{\unsure}[2][1=]{\mytodo{\todo[linecolor=blue,backgroundcolor=blue!25,bordercolor=blue,#1]{#2}}}
\newcommandx{\change}[2][1=]{\mytodo{\todo[linecolor=cyan,backgroundcolor=cyan!25,bordercolor=cyan,#1]{#2}}}
\newcommandx{\info}[2][1=]{\mytodo{\todo[linecolor=green,backgroundcolor=green!25,bordercolor=green,#1]{#2}}}
\newcommandx{\improve}[2][1=]{\mytodo{\todo[linecolor=red,backgroundcolor=red!25,bordercolor=red,#1]{#2}}}
\newcommandx{\thiswillnotshow}[2][1=]{\todo[disable,#1]{#2}}
\newcommand{\tableoftodos}{\ifnum\todotoggle=1 \listoftodos[Comments/To Do's] \fi}

\usetikzlibrary{positioning,snakes,mindmap,calc,fit,arrows,shapes}

\tikzstyle{box} = [shape=rectangle,draw=black,thick]

\newcolumntype{L}[1]{>{\raggedright\let\newline\\\arraybackslash\hspace{0pt}}m{#1}}
\newcolumntype{C}[1]{>{\centering\let\newline\\\arraybackslash\hspace{0pt}}m{#1}}
\newcolumntype{R}[1]{>{\raggedleft\let\newline\\\arraybackslash\hspace{0pt}}m{#1}}

\newcommand{\myignore}[1]{}

\newcommand{\TealBlue}[1]{\textcolor{TealBlue}{#1}}
\newcommand{\Peach}[1]{\textcolor{Peach}{#1}}
\newcommand{\Cyan}[1]{\textcolor{cyan}{#1}}
\newcommand{\Red}[1]{\textcolor{red}{#1}}
\newcommand{\Navy}[1]{\textcolor{Blue}{#1}}
\newcommand{\Blue}[1]{\textcolor{Blue}{#1}}
\newcommand{\Green}[1]{\textcolor{OliveGreen}{#1}}
\newcommand{\Black}[1]{\textcolor{black}{#1}}
\newcommand{\White}[1]{\textcolor{white}{#1}}
\newcommand{\Code}[1]{\Cyan{\texttt{#1}}}
\newcommand{\Gray}[1]{\textcolor{gray}{#1}}
\newcommand{\Magenta}[1]{\textcolor{magenta}{#1}}

\newcommand{\tTealBlue}[1]{\ifnum\todotoggle=1\TealBlue{#1}\else{#1}\fi}
\newcommand{\tPeach}[1]{\ifnum\todotoggle=1\Peach{#1}\else{#1}\fi}
\newcommand{\tCyan}[1]{\ifnum\todotoggle=1\Cyan{#1}\else{#1}\fi}
\newcommand{\tRed}[1]{\ifnum\todotoggle=1\Red{#1}\else{#1}\fi}
\newcommand{\tNavy}[1]{\ifnum\todotoggle=1\Navy{#1}\else{#1}\fi}
\newcommand{\tBlue}[1]{\ifnum\todotoggle=1\Blue{#1}\else{#1}\fi}
\newcommand{\tGreen}[1]{\ifnum\todotoggle=1\Green{#1}\else{#1}\fi}
\newcommand{\tBlack}[1]{\ifnum\todotoggle=1\Black{#1}\else{#1}\fi}
\newcommand{\tWhite}[1]{\ifnum\todotoggle=1\White{#1}\else{#1}\fi}
\newcommand{\tCode}[1]{\ifnum\todotoggle=1\Code{#1}\else{#1}\fi}
\newcommand{\tGray}[1]{\ifnum\todotoggle=1\Gray{#1}\else{#1}\fi}
\newcommand{\tMagenta}[1]{\ifnum\todotoggle=1\Magenta{#1}\else{#1}\fi}

\newcommand{\inbrace}[1]{\left \{ #1 \right \}}
\newcommand{\inparen}[1]{\left ( #1 \right )}

\newcommand{\inangle}[1]{\left \langle #1 \right \rangle}
\newcommand{\infork}[1]{\left \{ \begin{matrix} #1 \end{matrix} \right .}

\newcommand{\inabs}[1]{\begin{vmatrix} #1 \end{vmatrix}}
\newcommand{\norm}[2]{\begin{Vmatrix} #2 \end{Vmatrix}_{#1}}

\newcommand{\ceil}[1]{\left \lceil #1 \right \rceil}

\newcommand{\set}[1]{\inbrace{#1}}
\newcommand{\setdef}[2]{\set{#1 : #2}}
\newcommand{\defeq}{\stackrel{\mathrm{def}}{=}}

\newcommand{\indicator}{\mathbf{1}}

\newcommand{\half}{\nicefrac{1}{2}}

\newcommand{\bit}{\set{0,1}}

\newcommand{\wtilde}[1]{\widetilde{#1}}
\newcommand{\what}[1]{\widehat{#1}}

\newcommand{\eps}{\varepsilon}

\newcommand{\bbE}{{\mathbb E}}

\newcommand{\bbR}{{\mathbb R}}

\newcommand{\ba}{\mathbf{a}}
\newcommand{\bb}{\mathbf{b}}

\newcommand{\bx}{\mathbf{x}}
\newcommand{\by}{\mathbf{y}}

\newcommand{\calB}{\mathcal{B}}
\newcommand{\calC}{\mathcal{C}}

\newcommand{\calF}{\mathcal{F}}

\newcommand{\poly}{\mathrm{poly}}
\newcommand{\supp}{\mathrm{supp}}

\newcommand{\sign}{\mathrm{sign}}

\newtheorem{theorem}{Theorem}
\newtheorem{defn}[theorem]{Definition}

\newtheorem{lem}[theorem]{Lemma}
\newtheorem{proposition}[theorem]{Proposition}

\newtheorem{remark}[theorem]{Remark}

\newtheorem{open}{Open Problem}

\newenvironment{proof-overview}{\noindent {\em Proof Overview.} \hspace*{1mm}}{\hfill $\Box$}
\newenvironment{proofof}[1]{\noindent {\em Proof of #1.} \hspace*{1mm}}{\hfill $\Box$}
\newenvironment{proof-sketch}{\noindent {\em Sketch of Proof.} \hspace*{1mm}}{\hfill $\Box$}
\newenvironment{proof-idea}{\noindent{\em Proof Idea.}\hspace*{1mm}}{\qed\bigskip}
\newenvironment{proof-attempt}{\noindent{\em Proof Attempt.}\hspace*{1mm}}{\qed\bigskip}

%% file: abstract.tex
\begin{abstract}

Unlike compressive sensing where the measurement outputs are assumed to be real-valued and have infinite precision, in {\em one-bit compressive sensing}, measurements are quantized to one bit, their signs. In this work, we show how to recover the support of sparse high-dimensional vectors in the one-bit compressive sensing framework with an asymptotically near-optimal number of measurements. We also improve the bounds on the number of measurements for approximately recovering vectors from one-bit compressive sensing measurements. Our results are universal, namely the same measurement scheme works simultaneously for all sparse vectors.  

Our proof of optimality for support recovery is obtained by showing an equivalence between the task of support recovery using 1-bit compressive sensing and a well-studied combinatorial object known as Union Free Families.
\end{abstract}

%% file: sec_intro.tex
\section{Introduction} \label{sec:intro}

The problem of recovering a sparse signal from a small number of measurements is a fundamental one in machine learning, statistics, and signal processing. When the measurements are linear, the process is called {\em compressive sensing}. Remarkable results from the last decade \cite{Don06, candes2006robust} have shown that it is possible to efficiently reconstruct sparse signals using only $\Theta(k \log(n/k))$ linear measurements. Here, $n$ is the ambient dimension of the input signal and $k$ is its sparsity. A particularly striking result in compressive sensing is that with high probability, a Gaussian matrix with $\Theta(k \log(n/k))$ rows can be used as the sensing matrix for all sparse inputs simultaneously and is in that sense universal.

A criticism of compressive sensing is that it assumes infinite-precision real-valued measurements. Quantization of measurement outputs to very low bit-rates cannot be modeled simply as additive noise with bounded norm. To address this issue, Boufounos and Baraniuk \cite{BB08} introduced the notion of {\em $1$-bit compressive sensing} where each measurement is quantized to a single bit, namely its sign. This quantization can be cheaply implemented in hardware and is robust to certain non-linear distortions \cite{Bou10}. One-bit compressive sensing is an active area of research (\textit{e.g.}, \cite{GNR10, yan2012robust, PlanV13, JLBB_Robust1bitCS, GJNN_1bitCS, zhang2014efficient, knudson2016one, dai2016noisy, liu2016one, shen2016one}).

Formally, in $1$-bit compressive sensing, given a sensing matrix $A$, measurements of a $k$-sparse\footnote{A vector is $k$-sparse if it has at most $k$ nonzero components.} signal $\bx \in \bbR^n$ are obtained by:
$$\by = \text{sign}(A \bx)$$
so that $\by$ is the vector of signs\footnote{To be precise, let $\sign(x) = x/|x|$ for nonzero $x$ and $\sign(0)=0$. Note that this seems to be returning more than $1$ bit. But observe that if we instead define $\overline{\sign}(x) = \infork{1 & x > 0 \\ -1 & x \le 0}$, then a measurement of $\sign(\inangle{\ba, \bx})$, can be simulated with two $\overline{\sign}$ measurements, namely using $\overline{\sign}(\inangle{\ba, \bx})$ and $\overline{\sign}(\inangle{-\ba, \bx})$.} of the coordinates of $A\bx$. We consider noiseless measurements. Note that all information about the magnitude of $\bx$ is lost by the sign operator, and we can only hope to reconstruct the normalized vector $\bx/\norm{2}{\bx}$ from $\by$. 

In this work, we primarily consider the problem of {\em support recovery} of sparse vectors using $1$-bit compressive sensing measurements. We focus on {\em universal} sensing matrices. This is commonly referred to as \emph{for all}, or as \emph{uniform bounds}. Universal sensing matrices have guarantees of the form, ``with high probability, for all signals, the algorithm succeeds'', which is in contrast to the general randomized setting where guarantees are slightly weaker, ``for each signal, with high probability, the algorithm succeeds''~\cite{gilbert2007one}. Our objective is to minimize the total number of measurements needed (i.e., the number of rows in the sensing matrix) and the running time of the recovery algorithm. Formally:
\begin{defn}[Support Recovery with 1-bit Compressed Sensing]
	A matrix $A \in \bbR^{m \times n}$ is a {\em 1-bit compressive sensing matrix} for support recovery of $k$-sparse vectors if there exists a recovery algorithm such that, for all $\bx \in \bbR^n$ satisfying $\norm{0}{\bx} \le k$, the algorithm on input $A\bx$ returns $\supp(\bx)$. 
\end{defn}

We will also consider the problem of approximate vector recovery using 1-bit compressive sensing, again focusing on the universality. Formally:
\begin{defn}[Approximate Vector Recovery with 1-bit Compressed Sensing]
	A matrix $A \in \bbR^{m \times n}$ is a {\em 1-bit compressive sensing matrix} for $\eps$-approximate vector recovery of $k$-sparse vectors if there exists a recovery algorithm such that, for all $\bx \in \bbR^n$ satisfying $\norm{0}{\bx} \le k$, the algorithm on input $A\bx$ returns $\hat{\bx}$ such that $\inabs{\frac{\bx}{\norm{2}{\bx}}-\frac{\hat{\bx}}{\norm{2}{\hat{\bx}}}}<\eps$. 
\end{defn}

\begin{table*}[t]
	\renewcommand*{\arraystretch}{1.3}
	\begin{center}
		\begin{tabular}{|c|c|c|c|}
			\hline
			Problem & Upper Bound & Lower Bound & Citation\\
			\hline
			\hline
			\multirow{3}{*}{\shortstack{Support Recovery for \\ $k$ sparse vectors in $\bbR^n$}}  & $O(k^3\log n)$ & -- & \cite{GJNN_1bitCS}  \\
			& -- & $\Omega(k \log \frac n k)$ & folklore\\
			& $O(k^2\log n)$ & $\Omega(k^2\log n / \log k)$ & This work \\
			\hline
			\multirow{4}{*}{\shortstack{Approximate Recovery for \\ $k$ sparse vectors in $\bbR^n$}} & $\wtilde{O}\inparen{\frac{k}{\eps}\log\frac nk}$ & -- & \cite{JLBB_Robust1bitCS, GJNN_1bitCS} \\
			&$\wtilde{O}\inparen{k^3\log \frac nk+\frac k\eps}$ & -- & \cite{GJNN_1bitCS} \\ 
			& -- & $\Omega\inparen{k\log \frac nk+\frac k\eps - k^{1.5}}$ & \cite{JLBB_Robust1bitCS}\\ 
			&$\wtilde{O}\inparen{k^2\log \frac nk+\frac k\eps}$ & $\Omega\inparen{k\log \frac nk + \frac k\eps}$ & This work \\ 
			\hline
		\end{tabular}
		\vspace{2mm}
		\caption{Summary of results on universal 1-bit compressive sensing}
		\label{table:summary}
	\end{center}
\end{table*}

\subsection{Our Results}

Our main contribution is to show nearly tight upper and lower bounds on the number of measurements needed for support recovery of $k$-sparse signals using 1-bit compressive sensing. We also provide some improvements on the bounds in approximate vector recovery. See Table~\ref{table:summary} for a summary of our results.\footnote{the notation $\wtilde{O}(f(n))$ stands for $O(f(n) \cdot \poly\log(f(n)))$ for any $\poly\log(f(n))$.}

\subsubsection{Support Recovery}

Previously, Gopi et al~\cite{GJNN_1bitCS} have shown a universal support recovery algorithm using $O(k^3 \log n)$ $1$-bit measurements with $O(nk \log n)$ running time. If universality is not a constraint, then~\cite{HB11, GNR10} show that $O(k \log n)$ measurements suffice.

Our main contribution is showing that $\Theta(k^2 \log n)$ is a nearly tight bound for the number of $1$-bit measurements needed for universal support recovery. Like in \cite{GJNN_1bitCS}, our arguments exploit the structure of {\em Union Free Set Families}~\cite{EFF82}. While~\cite{GJNN_1bitCS} uses Union Free Families to recover non-negative sparse vectors, we observe that a strengthened version of these set families can in fact be used to recover all sparse vectors.  Moreover, we prove that any $1$-bit compressive sensing matrix for support recovery can be converted into a Union Free Family, thus deepening the connection between the two notions. Formally, we obtain the following upper and lower bounds:

\begin{restatable}{theorem}{upperbound}\label{thm:upperbound} {\em (Upper bound for Support Recovery)}
There exists a $1$-bit compressive sensing matrix $A \in \bbR^{m \times n}$ for support recovery of $k$-sparse signals that uses $m = O(k^2 \log n)$ measurements. Moreover, the recovery algorithm runs in time $O(n k \log n)$.
\end{restatable}

\vspace*{-0.5cm}

\begin{restatable}{theorem}{lowerbound}\label{thm:lowerbound} {\em (Lower bound for Support Recovery)}
	Let $A \in \bbR^{m \times n}$ be such that the map $\psi_A : \bbR^n \to \bit^m$, given by $\psi_A(\bx) \defeq \sign(A\bx)$ satisfies $\psi_A(\bx_1) \ne \psi_A(\bx_2)$ whenever $\norm{0}{\bx_1}, \norm{0}{\bx_2} \le k$ and $\supp(\bx_1) \ne \supp(\bx_2)$. Then, $m = \Omega(k^2 \log n / \log k)$.
\end{restatable}

\paragraph{Comparison to Group Testing.}
We remark that quantitatively similar results were known previously in the context of Group Testing \cite{Dorfman:43}, which in the language of $1$-bit compressive sensing, corresponds to the setting where the $k$-sparse signals have entries in $\bit$, and the measurements are restricted to be non-negative. Indeed, these results are obtained by showing a tight connection between Group Testing and Union-Free Families (also known as $k$-disjunct families). Group Testing has been an active research topic with a vast literature (See for eg, \cite{BBTK:96, DH:00, HD:00, AtiaS12, CaiJBJ13, chan2014non, Mazumdar16} and references therein).

Our contributions are as follows: (i) In Theorem~\ref{thm:upperbound}, we use a strengthened notion of Union-Free Families, to obtain a better upper bound for support recovery of arbitrary $k$-sparse signals in $\bbR^n$; while surprisingly still using measurements vectors with entries in $\bit$. (ii) In Theorem~\ref{thm:lowerbound}, the lower bound we obtain is incomparable to the lower bound in the Group Testing problem. Our lower bound is stronger in the sense that it applies even when the measurements are arbitrary real vectors (instead of just non-negative), whereas it is weaker in the sense that the lower bound applies to measurements that can recovery the support for all $k$-sparse signals in $\bbR^n$ (instead of only $0$-$1$ signals).

\subsubsection{Approximate Vector Recovery}

A number of papers have obtained bounds for approximate vector recovery~\cite{PlanV13, JLBB_Robust1bitCS, GJNN_1bitCS}. The current universal 1-bit compressive sensing algorithms require \singcol{$\min\{\wtilde{O}(\frac{k}{\eps} \log \frac nk), \wtilde{O}(k^3\log
\frac nk + \frac k\eps)\}$}{\[ \min\set{\wtilde{O}\inparen{\frac{k}{\eps} \log \frac nk}, \wtilde{O}\inparen{k^3\log
\frac nk + \frac k\eps}} \]}measurements. \cite{JLBB_Robust1bitCS} also proved a lower bound of $\Omega(k\log \frac nk+ \frac k\eps - k^{3/2})$ measurements.\footnote{Strictly speaking, the lower bound of $\Omega(k \log \frac nk)$ is folklore (we provide a proof in Section~\ref{sec:approx} for completeness), and \cite{JLBB_Robust1bitCS} showed a lower bound of $\Omega(\frac{k}{\eps} - k^{1.5})$.} As a function of $\eps$, the second half of the bound is helpful only when $\eps<1/\sqrt k$.

As a corollary to Theorem~\ref{thm:upperbound}, we can improve the upper bound term of $\wtilde{O}(k^3\log
\frac nk + \frac k\eps)$ to $\wtilde{O}(k^2\log n + \frac k\eps)$. Moreover, in Section~\ref{sec:approx} we also improve the lower bound to $\Omega(k\log \frac nk+ \frac k\eps)$, which holds for all $\eps>0$.


\begin{restatable}{cor}{ubapprox}\label{cor:ub-approx}{\em (Improved Upper Bound for Approximate Recovery)}
	There exists a $1$-bit compressive sensing matrix $A \in \bbR^{m \times n}$ for $\eps$-approximate recovery of $k$-sparse signals that uses $m = \wtilde{O}\inparen{k^2 \log n + \frac{k}{\eps}}$ measurements.
\end{restatable}

\vspace*{-0.5cm}


\begin{restatable}{theorem}{lbapprox}\label{thm:lb-approx} {\em (Improved Lower Bound for Approximate Recovery)}
	The number of measurements for $\eps$-approximate recovery using 1-bit
	compressive sensing is at least $\Omega\inparen{k\log\frac nk+\frac k\eps}$.
\end{restatable}

%% file: sec_upperbound.tex
\section{Upper Bound for Support Recovery} \label{sec:upperbound}

In this section we prove Theorem~\ref{thm:upperbound}. Gopi et
al~\cite{GJNN_1bitCS} present two techniques to obtain 1-bit
compressive sensing matrices for support recovery of $k$-sparse
signals. The first technique is based on Union-Free-Families (UFF) to
solve support recovery using only $O(k^2 \log n)$
measurements. However, this technique works only when the signals are
non-negative. In order to handle all real-valued signals, they propose
a technique based on expanders that uses $O(k^3 \log n)$
measurements. This expander based technique can be interpreted as implicitly constructing a generalization of
UFFs called \emph{Robust-UFF} (Definition~\ref{def:robust-uff}). This construction is able to handle
all real signals, albeit with an additional multiplicative factor of $k$ in the number of measurements. Our upper bound
uses Robust-UFFs constructed directly using the probabilistic method instead of going via expanders, thereby leading to a 1-bit compressive sensing matrix for support recovery using only $O(k^2\log n)$ measurements.

\begin{defn}[Union Free Family]\label{def:uff}
	A family of sets $\calF = \set{B_1, B_2, \ldots, B_n}$, where each $B_i \subseteq [m]$\footnote{$[m] := \set{1, \ldots, m}$} is an {\em $(n,m,k)$-UFF} if the following holds: for all distinct $j_0, j_1, \ldots, j_k \in [n]$, it is the case that $B_{j_0} \not\subseteq (B_{j_1} \cup B_{j_2} \cup \cdots \cup B_{j_k})$.
\end{defn}

\begin{defn}[Robust Union Free Family]\label{def:robust-uff}
	A family of sets $\calF = \set{B_1, B_2, \ldots, B_n}$, where each $B_i \subseteq [m]$ is an {\em $(n,m,d,k,\alpha)$-Robust-UFF} if the following holds: for all distinct $j_0, j_1, \ldots, j_k \in [n]$, it is the case that $|B_{j_0} \cap (B_{j_1} \cup B_{j_2} \cup \cdots \cup B_{j_k})| < \alpha |B_{j_0}|$ and $|B_j| = d$ for every $j \in [n]$.
\end{defn}

An easy application of the probabilistic method shows the existence of Robust-UFFs with certain desirable parameters, as done in \cite{deWolf_RobustUFF}.

\begin{lem}[Existence of Robust-UFF \cite{deWolf_RobustUFF}] \label{lem:Robust-UFF_exist}
	There exists an $(n,m,d,k,\alpha)$-Robust-UFF $\calF$ with parameters satisfying $m = O\inparen{\frac{k^2 \log n}{\alpha^2}}$ and $d = O\inparen{\frac{k \log n}{\alpha}}$.
\end{lem}


\begin{remark}
Union Free Families (UFF) are a special case of Robust-UFF when
$\alpha=1$, namely $|B_{j_0} \cap (B_{j_1} \cup B_{j_2} \cup \cdots \cup B_{j_k})| < |B_{j_0}|$.
\end{remark}


\subsection*{Support recovery from Robust-UFFs}

We are now ready to prove Theorem~\ref{thm:upperbound} \singcol{}{(restated below for convenience)}by constructing a suitable 1-bit compressive sensing matrix.

\ifnum\arxiv=1
\upperbound*
\fi

\begin{algorithm}[t]
	\KwIn{$\bb = \sign(A \bx^*)$}
	
	$\what{S} \gets \emptyset$
	
	\For{$j \in [n]$}{
		\If{$\inabs{B_j \cap \supp(\mathbf{\bb})} > d/2$}{
			$\what{S} \gets \what{S} \cup \set{j}$	
		}	
	}
	
	\KwOut{$\what{S}$} 
	
	\caption{1-bit Compressed Sensing for Support Recovery from Robust UFFs}
	\label{alg:OBCS}
\end{algorithm}

\ifnum\arxiv=1 \vspace*{-0.5cm} \fi
\ifnum\arxiv=0
\begin{proofof}{Theorem~\ref{thm:upperbound}}
\else
\begin{proof}
\fi
	Starting from any $(n,m,d,k,\half)$-Robust-UFF $\calF = \set{B_1, \ldots, B_n}$, we construct a compressive sensing matrix $A \in \bit^{m \times n}$ as follows: $A_{i,j} = \indicator_{(i \in B_j)}$. From Lemma~\ref{lem:Robust-UFF_exist}, we have that such a Robust-UFF exists with $m = O(k^2 \log n)$ and $d = O(k \log n)$. On receiving input $\bb = A\bx^*$, the support recovery algorithm proceeds as follows: Include $j$ into set $\what{S}$ if and only if at least half of the measurements corresponding to set $B_j$ are non-zero. See Algorithm~\ref{alg:OBCS} for a more detailed pseudo-code.
	
	{\bf Correctness.} Suppose the $k$-sparse vector is supported on coordinates $x_1, \ldots, x_k$ (the proof works similarly for other supports). Firstly for any $j \notin [k]$, we have that $|B_j \cap (B_1 \cup \cdots \cup B_k)| < |B_j|/2 = d/2$ from the definition of Robust-UFF. Thus, irrespective of the values of $x_1, \ldots, x_k$, the measurement outcomes corresponding to $B_j \setminus (B_1 \cup \cdots \cup B_k)$ will always be zero. Since more than $d/2$ of the measurements in $B_j$ are zero, $j$ will not be included in the set $\what{S}$. Next, consider any $j \in [k]$. Again, from the definition of Robust-UFF, we have that $\inabs{B_j \cap \inparen{\bigcup_{i \in [k], i \ne j} B_i}} < |B_j|/2 = d/2$. Thus, irrespective of the values of $x_1, \ldots, x_k$, the measurement outcomes corresponding to $B_j \setminus \inparen{\bigcup_{i \in [k], i \ne j} B_i}$ will be non-zero. Since more than $d/2$ of the measurements in $B_j$ are non-zero, $j$ will be included in the set $\what{S}$.
	
	{\bf Efficiency.} It easy to see that each iteration of the algorithm takes $O(k \log n)$ time, and hence overall the algorithm runs in $O(nk \log n)$ time. Note that, here we are not accounting for the time needed to construct the matrix $A$ which is part of pre-processing.
\ifnum\arxiv=1
\end{proof}
\else
\end{proofof}
\fi

%% file: sec_lowerbound.tex
\section{Lower Bound for Support Recovery} \label{sec:lowerbound}

In this section we prove Theorem~\ref{thm:lowerbound}. We prove this lower bound in two steps,
\begin{enumerate}
\item we show that $1$-bit compressive sensing implies the existence of a Union Free Family with similar parameters, \singcol{}{and}
\item we use known upper bounds on the size of Union Free Families to prove our lower bound.
\end{enumerate}

We start with the second point, for which we simply use the upper bound on the size of UFFs due to F\"{u}redi~\cite{Furedi_CoverFree}.

\begin{lem}[Upper bound on Union-Free Families \cite{Furedi_CoverFree}] \label{lem:furedi}
	Let $\calF = \set{B_1, \ldots, B_n}$ be a family of subsets of $[m]$, and $k \ge 2$, such that for any $j_0, j_1, \ldots, j_k$, it holds that $B_{j_0} \not\subseteq B_{j_1} \cup B_{j_2} \cup \cdots \cup B_{j_k}$. Then,
	\ifnum\arxiv=1
	$$n \le k + \binom{m}{t} \quad \text{ where, } t = \ceil{\frac{m-k}{\binom{k+1}{2}}}.$$
	\else
	$n \le k + \binom{m}{t}$ where, $t = \ceil{\frac{m-k}{\binom{k+1}{2}}}$.
	\fi
	This implies $m \ge \Omega(k^2 \log n/\log k)$. 
\end{lem}

For ease of presentation, we first prove a lower bound on the number of measurements for exact support recovery using only non-negative measurements, namely in the following theorem the entries of $A$ are non-negative (note that the compressive sensing matrix obtained in the proof of Theorem~\ref{thm:upperbound} in fact had only $0$-$1$ entries). More strongly, our lower bound works even when the matrix has to recover the support for only $0$-$1$ vectors. We remark that this lower bound was already known in the context of Combinatorial Group Testing~\cite{d1982bounds, dyachkov1989superimposed, Furedi_CoverFree}. We still present this proof first as it serves as a natural segue into our main lower bound.

\begin{theorem}[Lower bound for non-negative measurements] \label{thm:LB_NNM}
	Let $A \in \bbR_{\ge 0}^{m \times n}$ be such that the map $\psi_A : \bit^n \to \bit^m$, given by $\psi_A(\bx) \defeq \sign(A\bx)$ satisfies $\psi_A(\bx_1) \ne \psi_A(\bx_2)$ whenever $\norm{0}{\bx_1}, \norm{0}{\bx_2} \le k$ and $\supp(\bx_1) \ne \supp(\bx_2)$. Then, $m = \Omega(k^2 \log n / \log k)$.
\end{theorem}
\begin{proof}
	Any algorithm for exact support recovery of non-negative $k$-sparse signals, which uses only positive measurements, can be converted into an $(n,m,k-1)$-UFF. Suppose $A$ is a matrix achieving support recovery for non-negative $k$-sparse signals with $m$ measurements. Let $B_1, B_2, \ldots, B_n \subseteq [m]$ be such that $B_j = \setdef{i}{A_{ij} > 0}$. Suppose for contradiction that $\calB = \set{B_1, \ldots, B_n}$ is not an $(n,m,k-1)$-UFF. Then there exists $j_{0}, j_1, \ldots, j_{k-1}$ such that $B_{j_0} \subseteq B_{j_1} \cup \ldots \cup B_{j_{k-1}}$. Let $\bx_1 = \indicator(\set{j_1, \ldots, j_{k-1}})$ and $\bx_2 = \indicator(\set{j_0, j_1, \ldots, j_{k-1}})$. It is easy to see that $\psi_A(\bx_1) = \psi_A(\bx_2)$, which is a contradiction. Thus, we conclude that $\calB$ is a $(n,m,k-1)$-UFF and hence from Lemma~\ref{lem:furedi}, we get that $m \ge \Omega(k^2 \log n / \log k)$.
\end{proof}

Remarkably, we use the same technique to prove our main lower bound, i.e. Theorem~\ref{thm:lowerbound}\singcol{, }{ (restated below for convenience),}on the number of measurements needed for exact support recovery using arbitrary linear threshold measurements. However, here we need to use that the algorithm returns the exact support for all $(\le k)$-sparse vectors in $\bbR^n$ and not just those in $\bit^n$.

\ifnum\arxiv=1
\lowerbound*
\fi

\ifnum\arxiv=0
\begin{proofof}{Theorem~\ref{thm:lowerbound}}
\else
\begin{proof}
\fi
	Let $A$ be a matrix for support recovery of non-negative $k$-sparse signals. Without loss of generality, assume that $-1 \le A_{ij} \le 1$ for all $i, j$ (since scaling $A$ by constants doesn't change the outcome of sign measurements). Similar to the proof of Theorem~\ref{thm:LB_NNM}, let $B_1, B_2, \ldots, B_n \subseteq [m]$ be such that $B_j = \setdef{i \in [m]}{A_{ij} \ne 0}$. Suppose for contradiction that $\calB = \set{B_1, \ldots, B_n}$ is not a $(n,m,k-1)$-UFF. Hence, there exists $j_{0}, j_1, \ldots, j_{k-1}$ such that $B_{j_0} \subseteq B_{j_1} \cup \cdots \cup B_{j_{k-1}}$. We now construct two $(\le k)$-sparse vectors $\bx_1, \bx_2 \in \bbR^n$ with different supports such that $\psi_A(\bx_1) = \psi_A(\bx_2)$.
	
	Let $\bx_1$ be a vector supported on $j_1, \ldots, j_{k-1}$ such that all indices of $A\bx_1$ in $B_{j_1} \cup \cdots \cup B_{j_{k-1}}$ are $\eps$-away from $0$ for some choice of $\eps$\footnote{We can first choose a random $\bx_1'$ supported on $j_1, \ldots, j_{k-1}$, which will be such that all indices of $A\bx_1'$ in $B_{j_1} \cup \cdots \cup B_{j_{k-1}}$ are non-zero. Now we can get $\bx_1$ by simply scaling up this random $\bx_1'$ by a suitably large constant.}. Let $\bx_2 = \bx_1 + \eps \cdot e_{j_0}$. Since we assumed that $-1 \le A_{ij} \le 1$ for all $i, j$, we have that $A\bx_2 - A\bx_1 = A \cdot (\eps e_{j_0})$ has all entries with magnitude at most $\eps$. Since all entries of $A\bx_1$ in $B_{j_1} \cup \cdots \cup B_{j_{k-1}}$ are $\eps$-away from $0$, and $B_{j_0} \subseteq B_{j_1} \cup \cdots \cup B_{j_{k-1}}$, we get that $\psi_A(\bx_1) = \psi_A(\bx_2)$ even though $\supp(\bx_1) \ne \supp(\bx_2)$. Note that both $\bx_1$ and $\bx_2$ are $(\le k)$-sparse, and hence we get a contradiction.
	
	Thus, we conclude that $\calB$ is a $(n,m,k-1)$-UFF and hence from Lemma~\ref{lem:furedi}, we get that $m \ge \Omega(k^2 \log n/\log k)$.
\ifnum\arxiv=1
\end{proof}
\else
\end{proofof}
\fi

\noindent Thus, with Theorem \ref{thm:lowerbound}, we get a nearly tight lower bound of $\Omega(k^2 \log n/\log k)$ on the number of measurements needed for support recovery, even if we assume that the signals are non-negative and the measurements are allowed to be arbitrary. This is nearly matching the upper bound obtained in Theorem~\ref{thm:upperbound}, where we have a measurement matrix with $O(k^2 \log n)$ rows and only $0$-$1$ entries, which can recover support exactly for all signals in $\bbR^n$.

We note that our lower bound proof requires that the compressive sensing matrix correctly recovers the support for signals with arbitrarily large \emph{condition
number}. The condition number (or dynamic range) of a signal $\bx = (\bx_1, \ldots,
\bx_n)$ is defined as \ifnum\arxiv=0$K_\bx = \frac{\max_{i:\bx_i\ne0} |\bx_i|}{\min_{i:\bx_i\ne0} |\bx_i|}$, \else
\[K_\bx = \frac{\max_{i:\bx_i\ne0} |\bx_i|}{\min_{i:\bx_i\ne0} |\bx_i|},\]\fi
which is the highest ratio of absolute values of non-zero components of the
signal.

Signals with bounded condition numbers are easier to handle and are also
robust to noise. For example, \cite{GNR10} considered the case of signals with bounded condition number (in addition to presence of noise), although their measurements work in the non-universal setting.
Obtaining bounds on the number of measurements required in the universal setting, as a function
of the condition number, is open, even in the absence of noise. Even
the case when the condition number is $1$ is open (see the discussion in Section \ref{sec:disc}).


%% file: sec_approx-recovery.tex
\section{Approximate vector recovery} \label{sec:approx}

\subsection{Upper Bound}
For the problem of approximate vector recovery, note as in Table~\ref{table:summary}, that the two known upper bounds are $\wtilde{O}(\frac k\eps\log \frac nk)$, and  $\wtilde{O}(k^3\log
\frac nk + \frac k\eps)$. We improve the second bound of $\wtilde{O}(k^3\log
\frac nk + \frac k\eps)$ by~\cite{GJNN_1bitCS}  to $\wtilde{O}(k^2\log n + \frac k\eps)$ in Corollary~\ref{cor:ub-approx}\singcol{.}{ (restated for convenience).}

\ifnum\arxiv=1 \ubapprox* \fi

\ifnum\arxiv=1
\begin{proof}
\else
\begin{proofof}{Corollary~\ref{cor:ub-approx}}
\fi
The upper bound of $\wtilde{O}(k^3\log \frac nk + \frac k\eps)$ in \cite{GJNN_1bitCS} is shown by recovering the support of the vector using $O(k^3\log \frac nk)$ measurements and subsequently using $\wtilde{O}(k/\eps)$ measurements to approximately recover the vector in $k$ dimensions (this is still non-adaptive because standard Gaussian measurements suffice to approximately recover the vector).

Instead, using our improved algorithm of Theorem~\ref{thm:upperbound}, we need only $O(k^2\log n)$ measurements for support recovery, thereby obtaining the overall bound.
\ifnum\arxiv=0
\end{proofof}
\else
\end{proof}
\fi

\subsection{Lower Bound}

A lower bound of $\Omega(k\log \frac nk+ \frac k\eps -
k^{3/2})$ measurements for $\eps<\frac1{\sqrt k}$ was shown in~\cite{JLBB_Robust1bitCS}. We prove the same bound for all values of $\eps$ up to a constant in Theorem~\ref{thm:lb-approx}\singcol{. }{ (restated below for convenience).}We essentially follow the approach of~\cite{JLBB_Robust1bitCS}, but unlike their lower bound, focus on only one set of $k$ coordinates, instead of all possible sparsity
patterns. Surprisingly, this gives us a simpler proof that improves the lower bound by getting rid of the $k^{3/2}$ term.

\ifnum\arxiv=1 \lbapprox* \fi
\ifnum\arxiv=1
\begin{proof}
\else
\begin{proofof}{Theorem~\ref{thm:lb-approx}}
\fi
The first term of $k\log\frac nk$ is folklore. Nevertheless, we present the proof here for completeness. Consider the set of all $k$-sparse vectors of unit norm that have each non-zero entry equal to $1/\sqrt k$. Using the Gilbert-Varshamov bound \cite{Gilbert52, Varshamov57}, within this set there is a subset of at least $M = {n \choose k} / {n \choose \eps k}$
elements such that for $u$, and $v$ in the set, their supports have intersection at most $(1-\eps)k$. This implies that $\|u-v\|_2\ge\Omega(\eps)$. Since $m$ sign measurements can give us only $m$ bits of information, this gives us that $2^m \ge M$. By Stirling's approximation, for any $\eps<0.5$, we have \ifnum\arxiv=0$m \ge \log M \ge \Omega\left(k\log \frac nk\right)$. \else\[m \ge \log M \ge \Omega\left(k\log \frac nk\right).\]\fi This shows the first term. For the second term, we use the following lemma.

\begin{lem}[cf. Lemma 1 in \cite{JLBB_Robust1bitCS}] \label{lem:HP-division}
Let $m \ge 2k$. Then $m$-hyperplanes in $k$-dimensions divides the
region into at most $2^k\binom{m}{k}$ regions. 
\end{lem}

We now use the following well known lower bound on an $\eps$-cover for
$S^{k-1}$ (this follows from a straight forward volume argument).

\begin{lem}[$\eps$-cover for $S^{k-1}$] \label{lem:eps-cover}
There exists a subset $\mathcal{C} \subseteq S^{k-1}$, and a constant $c>0$ such that, $|\mathcal{C}| \ge
\inparen{\frac{c}{\eps}}^k$, and for all $\bx, \by \in \mathcal{C}$, it holds
that $\norm{2}{\bx - \by} \ge \eps$. 
\end{lem}

Now consider any 1-bit compressive sensing matrix with $m$
rows. They will correspond to $m$ hyperplanes. To reconstruct
all the vectors in $\mathcal{C}$, each entry of the $\mathcal{C}$ must
lie in a different region that the $m$ hyperplanes slice $S^{k-1}$
into. This in turn requires $2^k \binom{m}{k} \ge
\inparen{\frac{c}{\eps}}^k$. Since $\binom{m}{k} <(me/k)^k$, we get $2em/k > c/\eps$, thereby proving the bound.

We remark that our analysis is very similar to that of~\cite{JLBB_Robust1bitCS}. The main difference is that instead of considering sparse signals as lying in $n$ dimensions simultaneously, which is a union of subspaces, we just consider the set of all signals lying in $k$ dimensions.
\ifnum\arxiv=0
\end{proofof}
\else
\end{proof}
\fi

\noindent Thus, combining our results with prior literature, the upper and lower bounds for $\eps$-vector recovery stand as follows:
\begin{itemize}
	\item Upper bound: $\min\set{\wtilde{O}\inparen{\frac k\eps \log \frac nk}, \wtilde{O}\inparen{k^2\log \frac nk + \frac k\eps}}$.
	\item Lower bound: $\Omega\inparen{k \log \frac{n}{k} + \frac k\eps}$.
\end{itemize}

%% file: sec_discussion.tex
\section{Open Problems} \label{sec:disc}

We point out three intriguing open problems that so far have resisted easy answers. The first problem is the support recovery problem for vectors with condition number $1$.

\begin{open} \label{open:supp-recovery-0-1} How many measurements are necessary and sufficient
for a universal algorithm that recovers the support of
all 0-1 vectors that are $k$-sparse using 1-bit compressive sensing?
\end{open}

It can be shown that $O(k^{3/2}\log n)$ random gaussian measurements
suffice to recover all 0-1 vectors. This requires a simple computation that also follows from Theorem 2 of~\cite{JLBB_Robust1bitCS}. On the other hand, the best known lower
bound is the trivial $\Omega(k \log(n/k))$ measurements.\\

\noindent Our second open problem is about the approximate vector recovery problem.

\begin{open} What is the correct complexity of $\eps$-approximate vector recovery using 1-bit compressive sensing? 
\end{open}

We know from Section~\ref{sec:approx} and \cite{JLBB_Robust1bitCS, GJNN_1bitCS}, that $\min\set{\tilde{O}(\frac k\eps \log \frac nk), \tilde{O}(k^2\log \frac nk + \frac k\eps)}$ is an upper bound and $\Omega\inparen{k \log \frac{n}{k} + \frac k\eps}$ is a lower bound. The bounds are within a constant factor of each other in the regime where $\eps < 1/(k \log \frac n k)$ and also in the regime where $\eps = \Theta(1)$. However, there is still a gap in the regime where $1 \gg \eps \gg 1/(k \log \frac n k)$.\\

\noindent Our final open problem is to obtain ``explicit'' constructions for Robust-UFFs with the parameters that we want.

\begin{open}\label{open:explicit_constructions}
	Obtain an efficient algorithm to construct $(n,m,d,k,\alpha)$-Robust-UFF with parameter $m = O\inparen{\frac{k^2 \log n}{\alpha^2}}$, that is, in time that is polynomial in $n$ and $k$.
\end{open}

See Appendix~\ref{sec:nw_designs} for some approaches to obtain explicit constructions of Robust-UFFs using explicit error correcting codes. Unfortunately, such approaches using known constructions of error correcting codes, seem to fall shy of achieving the parameters we want.


%% file: sec_NWdesigns.tex
\section{Towards Explicit Constructions of Robust-UFFs} \label{sec:nw_designs}

In this section, we describe some attempts to get explicit constructions of Robust-UFFs using explicit constructions of error correcting codes. 

\begin{proposition}
	If $\calF$ is a $(n,m,d,1,\alpha/k)$-Robust-UFF, then $\calF$ is also a $(n,m,d,k,\alpha)$-Robust-UFF.
\end{proposition}
\begin{proof}
	Let $\calF = \set{B_1, \ldots, B_n}$, where $B_i \subseteq [m]$ with $|B_i| = d$ for all $i \in [n]$. From the definition of $(n,m,d,1,\alpha/k)$-Robust-UFF, we have that for any $i \ne j$, it holds that $|B_i \cap B_j| < \alpha d/k$. Thus, we also get that for any $i_0, i_1, \cdots, i_k$, it holds that,
	\[ |B_{i_0} \cap (B_{i_1} \cup \cdots \cup B_{i_k})| \le \sum_{j=1}^{k} |B_{i_0} \cap B_{i_j}| < \alpha d\;. \]
	Thus, $\calF$ is also a $(n,m,d,k,\alpha)$-Robust-UFF.
\end{proof}

\noindent We note that Lemma~\ref{lem:Robust-UFF_exist} implies the existence of $(n,m,d,1,\alpha/k)$-Robust-UFF with roughly the same parameters as that for $(n,m,d,k,\alpha)$-Robust-UFF. So in terms of probabilistic constructions at least, we aren't making our task too difficult by focusing on the $k=1$ case.

Thus, it suffices to construct Robust-UFFs with parameter $k = 1$, which are well studied under the name of  Nisan-Wigderson designs \cite{nisan1994hardness}. In particular, it is known that $(n,m,d,1,\alpha/k)$-Robust-UFF can be constructed in time $n^{O(k^2)}$ with the parameters of our interest, namely $m = O\inparen{(k^2 \log n)/\alpha^2}$ and $d = O\inparen{(k \log n)/\alpha}$ (cf. \cite{trevisan01extractors}). We now describe how explicit constructions of error correcting codes, immediately give rise to explicit constructions of Nisan-Wigderson designs.

\begin{defn}
	$\calC \subseteq [q]^d$ is an error correcting code with (relative) distance $(1-\delta)$ and rate $r$ if
	\begin{itemize}
	\item for any $c_1 \ne c_2 \in \calC$, the distance $\Delta(c_1, c_2) \ge (1-\delta)d$,\\ where $\Delta(c_1, c_2) \defeq \#\setdef{i \in [d]}{c_1(i) \ne c_2(i)}$.
	\item $r = \frac{\log |\calC|}{d \log q}$, or equivalently, $|\calC| = q^{rd}$.
	\end{itemize}
\end{defn}

\begin{proposition}[Robust-UFFs from error correcting codes]
	Given an error correcting code $\calC \subseteq [q]^d$ with distance $(1-\delta)$ and rate $r$, it is possible to construct a $(n,m,d,1,\delta)$-Robust-UFF, with parameters, $n = |\calC| = q^{rd}$ and $m = q \cdot d = \frac{q \log n}{r \log q}$.
\end{proposition}
\begin{proof}
	We construct a $(n,m,d,1,\delta)$-Robust-UFF $\calF$ from such an error correcting code $\calC$ as follows. Consider a universe $[d] \times [q]$. For every codeword $c \in \calC$, we include the subset $S_c = \setdef{(i,c(i))}{i \in [d]}$. It is easy to see that $n = |\calC| = q^{rd}$ and $m = q \cdot d$. The only thing to verify is that for $c_1 \ne c_2 \in \calC$, it holds that $|S_{c_1} \cap S_{c_2}| \le \delta d$. This holds because $|S_{c_1} \cap S_{c_2}| = \#\setdef{(i,\sigma)}{i \in d, \ c_1(i) = c_2(i) = \sigma} = d - \Delta(c_1, c_2) \le \delta d$.
\end{proof}

Unfortunately, this approach of using error correcting codes falls shy of achieving the parameters we want. For example, it is known that Algebraic-Geometry codes of distance $(1-\Theta(1/\sqrt{q}))$ with rate $\Theta(1/q)$ exist. Setting $q = O(k/\alpha)$, we get a $(n,m,d,1,\alpha/k)$-Robust-UFF with parameter $m = \frac{q \log n}{r \log q} = \wtilde{O}\inparen{\frac{k^3 \log n}{\alpha^3}}$. On the other hand, Reed-Solomon codes with distance $(1-\delta)$ and rate $\delta$ exist (by setting $d = q \ge 1/\delta$). Setting $\delta = \alpha/k$ and $q$ s.t. $n = q^{\delta q}$, we get a $(n,m,d,1,\alpha/k)$-Robust-UFF with parameter $m = \frac{q \log n}{r \log q} = \wtilde{O}\inparen{\frac{k^2 \log^2 n}{\alpha^2}}$.

In order to obtain $m = O\inparen{\frac{k^2 \log n}{\alpha^2}}$, we would require an error correcting code with distance $(1-\delta)$ and rate $r \ge (\delta^2 q / \log q)$ (where $\delta = \alpha/k$). We are not aware of even probabilistic constructions of error correcting codes which satisfy these constraints on the parameters.

%% file: main_arxiv.bbl
\begin{thebibliography}{BBTK96}

\bibitem[AS12]{AtiaS12}
G.K. Atia and V.~Saligrama.
\newblock Boolean compressed sensing and noisy group testing.
\newblock {\em Information Theory, IEEE Transactions on}, 58(3):1880--1901,
  March 2012.

\bibitem[BB08]{BB08}
Petros~T Boufounos and Richard~G Baraniuk.
\newblock 1-bit compressive sensing.
\newblock In {\em Information Sciences and Systems, 2008. CISS 2008. 42nd
  Annual Conference on}, pages 16--21. IEEE, 2008.

\bibitem[BBTK96]{BBTK:96}
D.J. Balding, W.J. Bruno, D.C Torney, and E.~Knill.
\newblock A comparative survey of non-adaptive pooling designs.
\newblock In Terry Speed and MichaelS. Waterman, editors, {\em Genetic Mapping
  and DNA Sequencing}, volume~81 of {\em The IMA Volumes in Mathematics and its
  Applications}, pages 133--154. Springer New York, 1996.

\bibitem[Bou10]{Bou10}
Petros~T Boufounos.
\newblock Reconstruction of sparse signals from distorted randomized
  measurements.
\newblock In {\em Acoustics Speech and Signal Processing (ICASSP), 2010 IEEE
  International Conference on}, pages 3998--4001. IEEE, 2010.

\bibitem[CJBJ13]{CaiJBJ13}
Sheng Cai, Mohammad Jahangoshahi, Mayank Bakshi, and Sidharth Jaggi.
\newblock {GROTESQUE:} noisy group testing (quick and efficient).
\newblock {\em CoRR}, abs/1307.2811, 2013.

\bibitem[CJSA14]{chan2014non}
Chun~Lam Chan, Sidharth Jaggi, Venkatesh Saligrama, and Samar Agnihotri.
\newblock Non-adaptive group testing: Explicit bounds and novel algorithms.
\newblock {\em IEEE Transactions on Information Theory}, 60(5):3019--3035,
  2014.

\bibitem[CRT06]{candes2006robust}
Emmanuel~J Cand{\`e}s, Justin Romberg, and Terence Tao.
\newblock Robust uncertainty principles: Exact signal reconstruction from
  highly incomplete frequency information.
\newblock {\em IEEE Transactions on information theory}, 52(2):489--509, 2006.

\bibitem[DH00]{DH:00}
Dingzhu Du and Frank~K. Hwang.
\newblock {\em Combinatorial Group Testing and Its Applications}.
\newblock Applied Mathematics. World Scientific, 2000.

\bibitem[Don06]{Don06}
David~L Donoho.
\newblock Compressed sensing.
\newblock {\em IEEE Transactions on information theory}, 52(4):1289--1306,
  2006.

\bibitem[Dor43]{Dorfman:43}
Robert Dorfman.
\newblock The detection of defective members of large populations.
\newblock {\em The Annals of Mathematical Statistics}, 14(4):436--440, 12 1943.

\bibitem[DR82]{d1982bounds}
Arkadii~Georgievich D'yachkov and Vladimir~Vasil'evich Rykov.
\newblock Bounds on the length of disjunctive codes.
\newblock {\em Problemy Peredachi Informatsii}, 18(3):7--13, 1982.

\bibitem[DRR89]{dyachkov1989superimposed}
AG~Dyachkov, VV~Rykov, and AM~Rashad.
\newblock Superimposed distance codes.
\newblock {\em Problems of Control and Information Theory-problemy Upravleniya
  i Teorii Informatsii}, 18(4):237--250, 1989.

\bibitem[DSXZ16]{dai2016noisy}
Dao-Qing Dai, Lixin Shen, Yuesheng Xu, and Na~Zhang.
\newblock Noisy 1-bit compressive sensing: models and algorithms.
\newblock {\em Applied and Computational Harmonic Analysis}, 40(1):1--32, 2016.

\bibitem[dW12]{deWolf_RobustUFF}
Ronald de~Wolf.
\newblock Efficient data structures from union-free families of sets.
\newblock 2012.

\bibitem[EFF82]{EFF82}
Paul Erd{\"o}s, Peter Frankl, and Zolt{\'a}n F{\"u}redi.
\newblock Families of finite sets in which no set is covered by the union of
  two others.
\newblock {\em Journal of Combinatorial Theory, Series A}, 33(2):158--166,
  1982.

\bibitem[F{\"{u}}r96]{Furedi_CoverFree}
Zolt{\'{a}}n F{\"{u}}redi.
\newblock On r-cover-free families.
\newblock {\em J. Comb. Theory, Ser. {A}}, 73(1):172--173, 1996.

\bibitem[Gil52]{Gilbert52}
E.~N. Gilbert.
\newblock A comparison of signalling alphabets.
\newblock {\em Bell System Technical Journal}, 31:504--522, 1952.

\bibitem[GNJN13]{GJNN_1bitCS}
Sivakant Gopi, Praneeth Netrapalli, Prateek Jain, and Aditya Nori.
\newblock One-bit compressed sensing: Provable support and vector recovery.
\newblock In {\em Proceedings of The 30th International Conference on Machine
  Learning}, pages 154--162, 2013.

\bibitem[GNR10]{GNR10}
Ankit Gupta, Robert Nowak, and Benjamin Recht.
\newblock Sample complexity for 1-bit compressed sensing and sparse
  classification.
\newblock In {\em Information Theory Proceedings (ISIT), 2010 IEEE
  International Symposium on}, pages 1553--1557. IEEE, 2010.

\bibitem[GSTV07]{gilbert2007one}
Anna~C Gilbert, Martin~J Strauss, Joel~A Tropp, and Roman Vershynin.
\newblock One sketch for all: fast algorithms for compressed sensing.
\newblock In {\em Proceedings of the thirty-ninth annual ACM symposium on
  Theory of computing}, pages 237--246. ACM, 2007.

\bibitem[HB11]{HB11}
Jarvis Haupt and Richard Baraniuk.
\newblock Robust support recovery using sparse compressive sensing matrices.
\newblock In {\em Information Sciences and Systems (CISS), 2011 45th Annual
  Conference on}, pages 1--6. IEEE, 2011.

\bibitem[JLBB13]{JLBB_Robust1bitCS}
Laurent Jacques, Jason~N. Laska, Petros~T. Boufounos, and Richard~G. Baraniuk.
\newblock Robust 1-bit compressive sensing via binary stable embeddings of
  sparse vectors.
\newblock {\em {IEEE} Transactions on Information Theory}, 59(4):2082--2102,
  2013.

\bibitem[KSW16]{knudson2016one}
Karin Knudson, Rayan Saab, and Rachel Ward.
\newblock One-bit compressive sensing with norm estimation.
\newblock {\em IEEE Transactions on Information Theory}, 62(5):2748--2758,
  2016.

\bibitem[LGX16]{liu2016one}
Wenhui Liu, Da~Gong, and Zhiqiang Xu.
\newblock One-bit compressed sensing by greedy algorithms.
\newblock {\em Numerical Mathematics: Theory, Methods and Applications},
  9(02):169--184, 2016.

\bibitem[Maz16]{Mazumdar16}
Arya Mazumdar.
\newblock Nonadaptive group testing with random set of defectives.
\newblock {\em IEEE Transactions on Information Theory}, 62(12):7522--7531, Dec
  2016.

\bibitem[ND00]{HD:00}
Hung~Q. Ngo and Ding-Zhu Du.
\newblock {A Survey on Combinatorial Group Testing Algorithms with Applications
  to {DNA} Library Screening}.
\newblock In {\em DIMACS Series in Discrete Mathematics and Theoretical
  Computer Science}, 2000.

\bibitem[NW94]{nisan1994hardness}
Noam Nisan and Avi Wigderson.
\newblock Hardness vs randomness.
\newblock {\em J. Comput. Syst. Sci.}, 49(2):149--167, 1994.

\bibitem[PV13]{PlanV13}
Yaniv Plan and Roman Vershynin.
\newblock Robust 1-bit compressed sensing and sparse logistic regression: {A}
  convex programming approach.
\newblock {\em {IEEE} Trans. Information Theory}, 59(1):482--494, 2013.

\bibitem[SS16]{shen2016one}
Lixin Shen and Bruce~W Suter.
\newblock One-bit compressive sampling via $\ell_0$ minimization.
\newblock {\em EURASIP Journal on Advances in Signal Processing}, 2016(1):71,
  2016.

\bibitem[Tre01]{trevisan01extractors}
Luca Trevisan.
\newblock Extractors and pseudorandom generators.
\newblock {\em J. {ACM}}, 48(4):860--879, 2001.

\bibitem[Var57]{Varshamov57}
R.~R. Varshamov.
\newblock Estimate of the number of signals in error correcting codes.
\newblock {\em Dokl. Acad. Nauk SSSR}, 117:739--741, 1957.

\bibitem[YYO12]{yan2012robust}
Ming Yan, Yi~Yang, and Stanley Osher.
\newblock Robust 1-bit compressive sensing using adaptive outlier pursuit.
\newblock {\em IEEE Transactions on Signal Processing}, 60(7):3868--3875, 2012.

\bibitem[ZYJ14]{zhang2014efficient}
Lijun Zhang, Jinfeng Yi, and Rong Jin.
\newblock Efficient algorithms for robust one-bit compressive sensing.
\newblock In {\em ICML}, pages 820--828, 2014.

\end{thebibliography}
